\documentclass[letterpaper, 10 pt, conference]{ieeeconf}
\usepackage{url}
\usepackage{graphicx}
\usepackage[format=plain,font=footnotesize,labelfont=bf,labelsep=period]{caption}
\usepackage{float}
\usepackage{bm}

\usepackage{amsmath}
\usepackage{amssymb}
\usepackage{amsthm}
\usepackage{mathtools}
\usepackage{color}
\usepackage{setspace}
\usepackage{cite}
\usepackage{multirow}

\newtheorem{theorem}{Theorem}

\newtheorem{lemma}{Lemma}
\theoremstyle{definition}
\newtheorem{definition}{Definition}
\newtheorem{remark}{Remark}
\theoremstyle{definition}
\newtheorem{assumption}{Assumption}
\newtheorem{proposition}{Proposition}

\DeclareMathOperator*{\argmin}{argmin}

\newcommand{\x}{\mathbf{x}}           
\newcommand{\uu}{\mathbf{u}}           
\newcommand{\uh}{\mathbf{u}_t}        
\newcommand{\vv}{\mathbf{v}}          
\newcommand{\qq}{\mathbf{q}}          
\newcommand{\xp}{\mathbf{x}_{\mathrm{p}}}   
\newcommand{\xph}{\hat{\x}_{\mathrm{p}}}   
\newcommand{\f}{\mathbf{f}}           
\newcommand{\phib}{\bm{\phi}}     
\newcommand{\bb}{\mathbf{b}}           
\newcommand{\kd}{\mathbf{k}_{\mathrm{d}}}        
\newcommand{\Scal}{\mathcal{S}}
\newcommand{\Se}{\mathcal{S}_e}
\newcommand{\Sx}{\mathcal{S}_x}
\newcommand{\Su}{\mathcal{S}_u}
\newcommand\norm[1]{\left\lVert#1\right\rVert}

\def\BibTeX{{\rm B\kern-.05em{\sc i\kern-.025em b}\kern-.08em
T\kern-.1667em\lower.7ex\hbox{E}\kern-.125emX}}

\IEEEoverridecommandlockouts
\allowdisplaybreaks
\begin{document} 

\title{\LARGE \bf 
Integral Control Barrier Functions with Input Delay: \\ Prediction, Feasibility, and Robustness 
}

\author{
  Adam K. Kiss, Ersin Da\c{s}, Tamas G. Molnar, and Aaron D. Ames
  \thanks{Adam K. Kiss is with the Department of Applied Mechanics, Faculty of Mechanical Engineering, Budapest University of Technology and Economics, Budapest, Hungary, {\tt\small kiss\_a@mm.bme.hu}.}%
  \thanks{Ersin Da\c{s} is with the Department of Mechanical, Materials, and Aerospace Engineering, Illinois Institute of Technology, Chicago, IL 60616, USA, {\tt\small edas2@illinoistech.edu}.}
  \thanks{Tamas G. Molnar is with the Department of Mechanical Engineering, Wichita State University, Wichita, KS 67260, USA, {\tt\small tamas.molnar@wichita.edu}.}%
  \thanks{Aaron D. Ames is with the Department of Mechanical and Civil Engineering, California Institute of Technology, Pasadena, CA 91125, USA, {\tt\small ames@caltech.edu}.}
  }

\maketitle
\thispagestyle{empty}

\begin{abstract}
Time delays in feedback control loops can cause controllers to respond too late, and with excessively large corrective actions, leading to unsafe behavior (violation of state constraints) and controller infeasibility (violation of input constraints).
To address this problem, we develop a safety-critical control framework for nonlinear systems with input delay using dynamically defined (integral) controllers.
Building on the concept of {\em Integral Control Barrier Functions} (ICBFs), we concurrently address two fundamental challenges: compensating the effect of delays, while ensuring feasibility when state and input constraints are imposed jointly. 
To this end, we embed predictor feedback into a dynamically defined control law to compensate for delays, with the predicted state evolving according to delay-free dynamics. 
Then, utilizing ICBFs, we formulate a quadratic program for safe control design.
For systems subject to simultaneous state and input constraints, we derive a closed-form feasibility condition for the resulting controller, yielding a compatible ICBF pair that guarantees forward invariance under delay. 
We also address robustness to prediction errors (e.g., caused by delay uncertainty) using tunable 
robust ICBFs.
Our approach is validated on an adaptive cruise control example with actuation delay.
\end{abstract}


\section{Introduction}
\label{sec:intro}

Control Barrier Functions (CBFs)~\cite{ames2017cbf} provide a principled approach to safety-critical control by enabling the design of controllers that enforce the forward invariance of a safe set in the state space.
This is typically implemented through safety filters, formulated as a quadratic program (QP) that minimally modifies a given nominal controller.
This framework has been applied to autonomous vehicles, robotics, and multi-agent systems~\cite{ames2019control, jiang2024robust, xiong2025robust, harms2025safe}.

One of the key challenges in applying CBFs to real systems is input delay.
Actuation, computation, and communication all introduce delays that cause the commanded input to take effect only after some time.
A controller unaware of this delay evaluates the CBF condition at the current state, while the effect of the input will only be realized later---this mismatch can cause safety violations.
Predictor feedback~\cite{Krsticbook2008} addresses this problem by forward-integrating the system over the delay interval to obtain the predicted state at which the CBF condition is evaluated~\cite{molnar2022safety,abel2024constrained}. 
In practice, however, the delay is often uncertain,
which may induce a prediction error that acts as a disturbance.
Robustness to this error can be achieved through integral quadratic constraints~\cite{seiler2021control,quan2023tube}, input-to-state safety~\cite{molnar2022input}, or robust CBFs~\cite{nanayakkara2025safety,dacs2026safe}, while online delay estimation may further reduce conservatism~\cite{kim2025minimizing}.

The standard CBF framework~\cite{ames2017cbf} considers
state constraints only,
which limits its applicability when input constraints (such as actuator bounds) or actuator dynamics are present.
This is especially important for systems with delays, because a delayed response may become more excessive with higher inputs.
Integral Control Barrier Functions (ICBFs)~\cite{ames2020integral} overcome input constraints by allowing the barrier function to depend on both the system state and the control input, so that input constraints and actuator dynamics can be encoded alongside the state constraint~\cite{dacs2024rollover}. 
ICBFs, however, have not yet been adapted to time delay systems.

Furthermore, ICBFs present another challenge.
For multiple constraints (both state and input limits), the combined QP may be infeasible even if each constraint individually admits a solution.
This phenomenon is well-known for CBF-based safety filters \cite{mousavi2025vertices}: when constraints are incompatible, the QP may have no solution, or undesired equilibria may appear on the safe set boundary~\cite{tan2024undesired}. 
The compatibility (or joint feasibility) of constraints has been studied for intersections of safe sets~\cite{molnar2023composing,cohen2025compatibility} and for explicit safety filter formulations~\cite{mestres2025explicit}.
Deriving a closed-form feasibility condition for ICBF-based safety filters with joint state and input constraints is still missing from the literature.

In this paper, we develop an ICBF framework for nonlinear control systems with input delay and joint state and input constraints.
First, we incorporate predictor feedback into the ICBF framework so that the safety filter operates on the predicted state, compensating for the delay.
Second, we derive a closed-form feasibility condition for the two-constraint QP arising from joint state and input limits,
constituting a compatible ICBF pair that guarantees feasibility and forward invariance under delay.
Third, we handle uncertain delay via a robust ICBF formulation based on~\cite{nanayakkara2025safety,dacs2026safe} that treats prediction errors as a disturbance, yielding safety guarantees that improve with prediction accuracy. We illustrate the framework on a safe adaptive cruise control example.

The rest of the paper is organized as follows.
Section~\ref{sec:background} revisits the ICBF framework and the two-constraint setup. 
Section~\ref{sec:icbfd} presents the main results: the predictor-based ICBF formulation, the feasibility condition, and tunable 
robustification. 
Section~\ref{sec:example} provides the adaptive cruise control simulation. 
Section~\ref{sec:conclusion} closes with conclusions.

\section{Background}
\label{sec:background}

We build on the framework of Control Barrier Functions (CBFs)~\cite{ames2017cbf,ames2019control} and Integral CBFs~\cite{ames2020integral}. We consider a general, not necessarily control-affine, nonlinear system:
\begin{equation}
  \dot{\x}(t) = \f(\x(t),\, \uu(t)),
  \label{eq:plant_nodelay}
\end{equation}
with continuously differentiable ${\f \!:\! \mathbb{R}^n \!\times\! \mathbb{R}^m \!\to\! \mathbb{R}^n}$. Throughout the paper, we omit the explicit time dependence when it is clear from the context for simplicity.

CBFs guarantee the forward invariance of the safe set 
\begin{equation*}
    \mathcal{C} := \{\x \in \mathbb{R}^n :h(\x)\ge 0\},
\end{equation*}
where the CBF
${h: \mathbb{R}^n \to \mathbb{R}}$
satisfies the condition
\begin{equation*}
\sup_{\uu\in\mathbb{R}^m}\left[\frac{\partial h}{\partial\x}(\x)\f(\x,\uu)+\alpha(h(\x))\right]\ge 0
  , \ \forall x \in \mathcal{C},
\end{equation*}
with an extended class-$\mathcal{K}$ function\footnote{Function  ${\alpha : [0,a) \to \mathbb{R}}$, ${a>0}$, is of class-$\mathcal{K}$ (${\alpha \in \mathcal{K}}$) if it is continuous, strictly increasing, and ${\alpha(0)=0}$.
Function ${\alpha : (-b,a) \to \mathbb{R}}$, ${a,b>0}$ is of extended class-$\mathcal{K}$ (${\alpha \in \mathcal{K}^{\rm e}}$) if it has the same properties.} ${\alpha \!\in\! \mathcal{K}^{\rm e}}$.
ICBFs extend this construction to dynamically defined control laws.
\subsection{Integral Control Barrier Functions}
\label{ssec:icbf}
We revisit the ICBF framework \cite{ames2020integral} for system  \eqref{eq:plant_nodelay}. 
Rather than an algebraic feedback controller ${\uu = \mathbf{k}(\x)}$, we consider a \emph{dynamically defined (integral) controller}:
\begin{equation}
  \dot{\uu}(t) = \phib(\x(t),\, \uu(t)) + \vv(t),
  \label{eq:ctrl_nodelay}
\end{equation}
where $\phib : \mathbb{R}^n \times \mathbb{R}^m \to \mathbb{R}^m$ is the continuously differentiable integral controller and
${\vv \in \mathbb{R}^m}$
is an auxiliary corrective input to be determined by the safety filter. For example, with a nominal controller ${\kd \!:\! \mathbb{R}^n \!\to\! \mathbb{R}^m}$, a common choice is \cite{ames2020integral}:
\begin{equation}
  \phib(\x, \uu) =
    \frac{\partial \kd}{\partial \x}(\x)\,\f(\x,\uu)
    + \frac{\alpha_\phi}{2}\left(\kd(\x) - \uu\right),
  \label{eq:phi_tracking}
\end{equation}
for $\alpha_\phi > 0$, which ensures that $\uu$ approaches $\kd(\x)$ exponentially along the trajectories of \eqref{eq:plant_nodelay}--\eqref{eq:ctrl_nodelay} with $\vv \equiv \mathbf{0}$.


We define safety using a continuously differentiable, input-dependent safety function ${h : \mathbb{R}^n \times \mathbb{R}^m \to \mathbb{R}}$ on the augmented state $(\x, \uu)$.
We assume that zero is a regular value of $h$, and consider
the associated safe set:
\begin{equation}
  \Scal := \{(\x,\uu) \in \mathbb{R}^n \times \mathbb{R}^m : h(\x,\uu) \geq 0\}.
  \label{eq:safeset}
\end{equation}

\begin{definition}[\!\!\cite{ames2020integral}]
\label{def:icbf}
Function $h:\mathbb{R}^n\times\mathbb{R}^m\to\mathbb{R}$ is an \emph{Integral Control Barrier Function} for
\eqref{eq:plant_nodelay} with the controller \eqref{eq:ctrl_nodelay} if there exists $\alpha\in\mathcal{K}^{\rm e}$ such that for all $(\x,\uu)\in\mathbb{R}^n\times\mathbb{R}^m$:
\begin{equation}
  \sup_{\vv\in\mathbb{R}^m}
  \left[\dot{h}(\x,\uu,\vv) + \alpha\big(h(\x,\uu)\big)\right] \ge 0.
  \label{eq:icbf_cond}
\end{equation}
\end{definition}

Let us define:
\begin{align}
  \!\!\bb(\x,\!\uu)  & \!:=\! \left ( \tfrac{\partial h}{\partial \uu}(\x,\!\uu) \right )^\top,
  \label{eq:ab_nodelay} \\[4pt]
   \!\!a(\x,\!\uu)  & \!:=\!
  - \tfrac{\partial h}{\partial \x}(\x,\!\uu) \f(\x,\!\uu)
  \!-\! \tfrac{\partial h}{\partial \uu}(\x,\!\uu) \phib(\x,\!\uu)
  \!-\! \alpha\big(h(\x,\!\uu)\big), \nonumber
\end{align}
so that
$\dot{h}(\x,\uu,\vv) + \alpha\big(h(\x,\uu)\big) \!=\! -a(\x,\uu)+\bb^\top \! (\x,\uu)\vv$. 
Condition~\eqref{eq:icbf_cond} is equivalent to the point-wise requirement:
\begin{equation}
  \bb(\x,\uu) = \mathbf{0} \;\implies\; a(\x,\uu) \leq 0,
  \ \forall\,(\x,\uu)\in\mathbb{R}^n\times\mathbb{R}^m.
  \label{eq:icbf_rd}
\end{equation}

\begin{theorem}[\!\!\!\!\cite{ames2020integral}]
\label{thm:icbf}
If $h$ is an ICBF for \eqref{eq:plant_nodelay} with the controller \eqref{eq:ctrl_nodelay}, then any locally Lipschitz controller ${\vv=\qq(\x,\uu)}$, with ${\qq : \mathbb{R}^n \times \mathbb{R}^m \to \mathbb{R}^m}$, satisfying:
\begin{equation}
  \bb^\top \! (\x,\uu) \vv \geq a(\x,\uu)
  \label{eq:icbf_safety_cond}
\end{equation}
for all ${(\x,\uu)\in\Scal}$,
renders $\Scal$ forward invariant (safe), i.e., ${(\x(0),\uu(0)) \in \Scal \implies (\x(t),\uu(t))\in \Scal}$, ${\forall\, t\ge 0}$.
\end{theorem}

A minimum-norm safety filter satisfying \eqref{eq:icbf_safety_cond} is given by
\begin{align}
  \qq(\x,\uu) = \argmin_{\vv \in \mathbb{R}^m} \quad & \norm{\vv}^2 \\[-3pt]
              \mathrm{s.t.} \quad & \bb^\top \! (\x,\uu) \vv \geq a(\x,\uu).\nonumber
  \label{eq:icbf_qp}
\end{align}
This QP admits the explicit solution
\begin{equation}
  \qq(\x,\uu) = \begin{cases}
    \tfrac{a(\x,\uu)}{\norm{\bb(\x,\uu)}^2}\,\bb(\x,\uu), & \text{if } a(\x,\uu) > 0, \\
    \mathbf{0}, & \text{if } a(\x,\uu) \leq 0,
  \end{cases}
  \label{eq:icbf_minnorm}
\end{equation}
obtained via the Karush-Kuhn-Tucker (KKT) conditions.

\subsection{State and Input Constraints}
\label{ssec:separable}

In practice, we often need to enforce two \emph{separate} constraints simultaneously: a state constraint and an input constraint. Let the safe sets for the state and input be
\begin{align}
  \Sx &:= \{(\x,\uu)\in\mathbb{R}^n \times \mathbb{R}^m : h_x(\x)\ge 0\},
  \label{eq:Sx}\\
  \Su &:= \{(\x,\uu)\in\mathbb{R}^n \times \mathbb{R}^m : h_u(\uu)\ge 0\},
  \label{eq:Su}
\end{align}
with continuously differentiable ${h_x : \mathbb{R}^n \!\to\! \mathbb{R}}$, ${h_u : \mathbb{R}^m \!\to\! \mathbb{R}}$.
Since $h_x$ depends only on $\x$,
it cannot serve as an ICBF.
Following \cite{ames2020integral}, we form the extended barrier:
\begin{equation}
  h_e(\x,\uu):=
  \frac{\partial h_x}{\partial \x}(\x)\f(\x,\uu)+\alpha_x(h_x(\x)),
  \label{eq:he}
\end{equation}
with $\alpha_x\in\mathcal{K}^{\rm e}$ and
\begin{equation}
  \Se:=\{(\x,\uu)\in\mathbb{R}^n\!\times\!\mathbb{R}^m : h_e(\x,\uu)\ge0\}.
  \label{eq:Se}
\end{equation}
Making $\Se$ forward invariant via $\vv$ guarantees
the forward invariance of ${\Se \cap \Sx}$, enforcing state constraints.

Let $h_e$ and $h_u$ be ICBFs.
Applying \eqref{eq:ab_nodelay} to $h_e$ and $h_u$ yields:
\begin{align}
  \bb_e(\x,\uu) &:= \left ( \tfrac{\partial h_e}{\partial\uu}(\x,\uu) \right)^\top,
  \label{eq:abe}\\
  a_e(\x,\uu)  &:= - \tfrac{\partial h_e}{\partial\x}(\x,\uu)\,\f(\x,\uu)
                         \nonumber\\
  & \;\;\quad - \tfrac{\partial h_e}{\partial\uu}(\x,\uu)\,\phib(\x,\uu)
                         - \alpha_e\big(h_e(\x,\uu)\big),
  \nonumber\\
  \bb_u(\x,\uu)    &:= \left (  \tfrac{\partial h_u}{\partial\uu}(\uu)\right)^\top,
  \label{eq:abu}\\
  a_u(\x,\uu)  &:= -\tfrac{\partial h_u}{\partial \uu}(\uu)\phib(\x,\uu)-\alpha_u(h_u(\uu)).
  \nonumber
\end{align}
One may enforce both constraints on the safety correction $\vv$ simultaneously, for example, by the QP:
\begin{align}
  \qq(\x,\uu) = \argmin_{\vv\in\mathbb{R}^{m}}
    &\quad \norm{\vv}^2
  \label{eq:combined_qp} \\[-3pt]
  \mathrm{s.t.} 
  &\quad \bb_e^\top(\x,\uu)\vv \geq a_e(\x,\uu), \nonumber\\
  &\quad \bb_u^\top(\x,\uu)\vv \geq a_u(\x,\uu). \nonumber
\end{align}
Even if each constraint individually admits a feasible $\vv$, the combined QP \eqref{eq:combined_qp} may be infeasible when both are imposed simultaneously. This naturally raises the question: \emph{when is the two-constraint QP guaranteed to be feasible?}
We address feasibility and forward invariance in Section~\ref{ssec:main_results}.

\section{Integral CBFs with Input Delay}
\label{sec:icbfd}

In this section, we present our main results: we extend the ICBF framework to systems with input delay.
We address joint state and input constraints, and robustness to uncertain prediction. The general case
is stated first, while the cases of perfect prediction
and no delay
follow as special cases.


Consider the system \eqref{eq:plant_nodelay} with a constant input delay $\tau>0$:
\begin{equation}
  \dot{\x}(t) = \f\big(\x(t),\, \uu(t{-}\tau)\big).
  \label{eq:plant_delay}
\end{equation}
We assume \eqref{eq:plant_delay} admits a unique solution $\x(t)$ for $t\ge 0$ given $\x(0)$ and a continuous input history ${\uh(\vartheta):=\uu(t+\vartheta)}$, ${\vartheta\in[-\tau,0)}$.
Due to the delay, the command $\uu(t)$ computed at time $t$ takes effect only at ${t+\tau}$, and during ${[t,\,t+\tau]}$ the system evolves under the history~$\uh$, which cannot be altered retroactively.
Evaluating the safety condition at $\x(t)$ is therefore inadequate for choosing~$\uu(t)$.

Instead, we require that the physical pair $(\x(t),\uu(t{-}\tau))$ -- the current state and the delayed input that is currently acting on the system --
satisfies the safety constraint:
\begin{equation}
    h_i \big( \x(t),\uu(t-\tau) \big) \geq 0,
\end{equation}
where $h_i$ may refer to a single constraint with $h$ or each of the two constraints with $h_e$ and $h_u$. With a slight abuse of notation, a barrier depending on a single variable is identified with its trivial extension to ${\x,\uu}$. 
The delayed system must evolve inside a safe set $\Scal_i$, so that
${(\x(t),\uu(t{-}\tau)) \in \Scal_i}$,
with
\begin{equation}
  \Scal_i := \big\{ (\x,\uu) \in\mathbb{R}^n \times \mathbb{R}^m : h_i(\x,\uu) \ge 0 \big\}.
  \label{eq:safeset_delay}
\end{equation}
We require the following regularity condition for $h_i$.

\begin{assumption} \label{assump:reg_value}
    Any nonpositive number is a regular value of $h_i$, that is,
    ${[\tfrac{\partial h_i}{\partial \x}(\x,\uu),\,\tfrac{\partial h_i}{\partial \uu}(\x,\uu)] \neq \mathbf{0}}$
    for all ${(\x,\uu) \in \mathbb{R}^n \times \mathbb{R}^m}$ satisfying ${h_i(\x,\uu) \leq 0}$.
\end{assumption}

\subsection{Predictor Feedback}
\label{ssec:system_pred}

To overcome the effect of delay,
we employ predictor feedback \cite{molnar2022safety}.
Assuming access to $\x(t)$ and the stored history~$\uh$,
the predicted state ${\xp(t):=\x(t+\tau)}$ is obtained by forward-integrating \eqref{eq:plant_delay} over the delay interval:
\begin{equation}
  \xp \!=\! \Psi(\tau, \x,\uh) \!:=\!
  \x + \!\int_0^\tau\!\! \f\big(\Psi(s, \x,\uh), \uh(s{-}\tau)\big) \mathrm{d}s,
  \label{eq:semiflow}
\end{equation}
where ${\uh(s-\tau)}$ is the known past input.
A key property \cite{Krsticbook2008} is that $\xp$ has delay-free dynamics:
${\dot{\x}_{\rm p}(t) = \f(\xp(t),\uu(t))}$.

The integral controller is designed with the predicted state:
\begin{equation}
  \dot{\uu}(t) = \phib(\xp(t),\,\uu(t)) + \vv(t),
  \label{eq:ctrl_delay}
\end{equation}
where $\phib$ is the integral controller and ${\vv = \qq(\xp,\uu)}$ is the auxiliary safety correction evaluated with the predicted state.
This controller is designed with two goals: ensuring that $\uu$ tracks a desired nominal law through $\phib$, and enforcing safety via $\vv$.
Because $\phib$ acts on the predicted state $\xp$ rather than the current state $\x$, delay compensation is embedded in the nominal dynamics.
For instance, for a static state feedback law ${\uu = \kd(\x)}$, a common choice in the delay setting is:
\begin{equation}
  \phib(\xp, \uu) =
    \frac{\partial \kd}{\partial \x}(\xp)\,\f(\xp,\uu)
    + \frac{\alpha_\phi}{2}\left(\kd(\xp) - \uu\right),
  \label{eq:phi_tracking_delay}
\end{equation}
for ${\alpha_\phi > 0}$, making $\uu$ converge to $\kd(\xp)$ exponentially.


Prediction requires the delay $\tau$ and the dynamics $\f$ to be accurately known.
However, in practice, the exact delay and dynamics are seldom known, which may lead to a {\em prediction error}: a mismatch between the true future state $\xp$ and its estimation $\xph$.
For example, if $\f$ is known but the controller uses a delay estimate $\hat{\tau}$, the estimated predicted state
\begin{equation}
  \xph(t) := \Psi\big(\hat{\tau},\x(t),\,\uh\big)
  \label{eq:pred_hat}
\end{equation}
differs from $\xp(t)=\Psi(\tau,\x(t),\uh)$ whenever $\hat{\tau}\neq\tau$.

While the physical system \eqref{eq:plant_delay} is driven by the input $\uu(t{-}\tau)$ and its dynamics remain undisturbed, the prediction error ${\xph - \xp}$ affects the integral control law in \eqref{eq:ctrl_delay}.
Namely, both the controller $\phib$ and the safety correction $\qq$ are evaluated at the uncertain $\xph$ instead of the true $\xp$:
\begin{equation}
  \dot{\uu}(t) = \phib\big(\xph(t),\,\uu(t)\big) + \qq(\xph(t),\,\uu(t)).
  \label{eq:ctrl_disturbed}
\end{equation}
The effect of the prediction mismatch can be viewed as a matched disturbance propagated into \eqref{eq:ctrl_delay}:
\begin{align}
  \dot{\uu}(t) & = \phib\big(\xp(t),\,\uu(t)\big) + \qq(\xp(t),\,\uu(t)) + \mathbf{d}(t), \\
  \mathbf{d} & := \phib(\xph,\uu) + \qq(\xph,\uu) - \phib(\xp,\uu) - \qq(\xp,\uu).
  \label{eq:disturbance_def}
\end{align}
We assume that the disturbance (i.e., the effect of prediction errors) stays bounded.


\begin{assumption}
\label{assump:delta}
${\norm{\mathbf{d}(t)} \leq \delta}$ for all ${t \geq 0}$, for some ${\delta \ge 0}$.
\end{assumption}

Knowing a bound ${\norm{\xph-\xp} \leq \delta_x}$ on the prediction error, the disturbance bound $\delta$ can be estimated, for example, as ${\delta = (L_\phi+L_\qq) \delta_x}$, where $L_\phi$ and $L_\qq$ are the Lipschitz constants of $\phib$ and $\qq$, respectively.
Notice that when ${\xph = \xp}$ (i.e., $\hat{\tau}=\tau$), we have $\mathbf{d}\equiv\mathbf{0}$ and ${\delta = 0}$.


To provide robust safety even in the presence of disturbance,
we adopt the robust CBF framework of~\cite{nanayakkara2025safety}.
For each barrier $h_i$,
we enforce a robust ICBF condition
\begin{equation} \label{eq:robust_cond}
    \bb_i^\top(\xp,\uu) \vv \geq a_i(\xp,\uu) + r_i(\xp,\uu),
\end{equation}
with the additional robustness term
\begin{equation}
\begin{aligned}
  r_i(\xp,\uu) & := \mu_i \big(h_i(\xp,\uu)\big) \norm{\bb_i(\xp,\uu)} \\
  & \quad + \sigma_i \big(h_i(\xp,\uu)\big) \norm{\bb_i(\xp,\uu)}^2,
  \label{eq:r}
\end{aligned}
\end{equation}
where
${\mu_i,\sigma_i:\mathbb{R} \!\to\! \mathbb{R}_{>0}}$ are continuously differentiable and monotonously decreasing.
A common choice is ${\mu_i(h_i)=\mu_0 e^{-\lambda h_i}}$, ${\sigma_i(h_i)=\sigma_0 e^{-\lambda h_i}}$ with ${\mu_0,\sigma_0,\lambda>0}$.
Using the tunability idea of \cite{Alan2022},
with $\mu_i$ and $\sigma_i$ we concentrate robustness near the safe set boundary (${h_i \approx 0}$), while leaving the controller unchanged far inside the set (${h_i \gg 0}$).

This finally leads to the following QP for ensuring robust safety in the presence of input delay and prediction errors:
\begin{align}
  \qq(\xp,\uu)  = \argmin_{\vv\in\mathbb{R}^{m}}
    &\quad \norm{\vv}^2
  \label{eq:combined_qp_delay} \\
  \mathrm{s.t.}
  &\quad \bb_i^\top(\xp,\uu) \vv \geq a_i(\xp,\uu) + r_i(\xp,\uu), \nonumber   
\end{align}
where the constraint is enforced for all ${i\in\{e,u\}}$ in case of two constraints (with $h_e$, $h_u$), while the index $i$ can be dropped for a single constraint (with $h$).
Next, we establish our main theoretical results underpinning this control design.

\subsection{Feasibility}
\label{ssec:main_results}
\label{def:tissf_icbfd}\label{def:compat_icbfd}\label{def:compat_icbf}

First, we establish feasibility conditions for the proposed controller.
In the case of a single constraint, we require that
there exists
${\alpha\in\mathcal{K}^{\rm e}}$ such that
for all ${(\xp,\uu) \in \mathbb{R}^n \times \mathbb{R}^m}$:
\begin{equation}
  \sup_{\vv\in\mathbb{R}^m}
  \Big[\dot{h}(\xp,\uu,\vv)
  +\alpha\big(h(\xp,\uu)\big)
  -r(\xp,\uu)\Big]\ge 0,
  \label{eq:tissf_cond_def}
\end{equation}
or equivalently:
\begin{equation}
  \bb(\xp,\uu)=\mathbf{0}
  \implies
  a(\xp,\uu)
  \le 0.
  \label{eq:tissf_rd}
\end{equation}
Note that ${r(\xp,\uu)\!=\!0}$ whenever ${\bb(\xp,\uu)\!=\!\mathbf{0}}$ based on \eqref{eq:r}, hence \eqref{eq:tissf_rd} reduces to the standard ICBF condition~\eqref{eq:icbf_rd}.
This implies that any ICBF is automatically a robust ICBF~\cite{nanayakkara2025safety}.

This condition is required for the feasibility of the single-constraint QP.
For the two-constraint case, we derive the feasibility condition
based on a form of Farkas' lemma.

\begin{lemma}[\!\!\cite{molnar2023composing,cohen2025compatibility}]
\label{lem:compat}
There exists $\vv\in\mathbb{R}^m$ satisfying two affine constraints, $\bb_e^\top\vv\ge a_e$ and $\bb_u^\top\vv\ge a_u$, if and only if
\begin{equation}
  \lambda_e\bb_e+\lambda_u\bb_u=\mathbf{0}
  \implies
  \lambda_e a_e+\lambda_u a_u\le 0,
  \label{eq:compat_farkas}
\end{equation}
for all ${\lambda_e,\lambda_u\ge 0}$.
\end{lemma}

Applying Lemma~\ref{lem:compat} to the two-constraint QP
yields the following easy-to-check condition.

\begin{proposition}
\label{prop:compat}
\textit{
The pair $(h_e,h_u)$ is a compatible robust ICBF pair and the QP in \eqref{eq:combined_qp_delay} is feasible if and only if
\begin{subequations}  \label{eq:robust_compat}
\begin{align}
    &\bb_e=\mathbf{0} \implies a_e 
    \le 0, \label{eq:compatibility_1} \\
    &\bb_u=\mathbf{0} \implies a_u 
    \le 0, \label{eq:compatibility_2} \\
    &\norm{\bb_e}\norm{\bb_u}+\bb_e^\top\bb_u=0 \implies \label{eq:compatibility_3} \\
    & \qquad (a_e+r_e)\norm{\bb_u}+(a_u+r_u)\norm{\bb_e}\le 0, \nonumber
\end{align}
\end{subequations}
where all quantities are evaluated at $(\xp,\uu)$.
}
\end{proposition}

\begin{proof}
We use Lemma~\ref{lem:compat}, by replacing $a_i$ with $a_i+r_i$.
We show that \eqref{eq:compat_farkas} holds for all cases where ${\lambda_e\bb_e+\lambda_u\bb_u=\mathbf{0}}$
with ${\lambda_e,\lambda_u\ge 0}$.
This means $\bb_e$ and $\bb_u$ are anti-parallel or zero (or ${\lambda_e=\lambda_u=0}$, for which \eqref{eq:compat_farkas} holds trivially).

\emph{Case $\bb_e=\mathbf{0}$, $\bb_u\ne\mathbf{0}$:}
${\lambda_e\bb_e+\lambda_u\bb_u=\mathbf{0}}$ holds if $\lambda_u=0$, so \eqref{eq:compat_farkas} requires $\lambda_e(a_e+r_e)\le 0$
for all $\lambda_e\ge 0$, giving~\eqref{eq:compatibility_1}.

\emph{Case $\bb_e\ne\mathbf{0}$, $\bb_u=\mathbf{0}$}
is symmetric and gives~\eqref{eq:compatibility_2}.

\emph{Case $\bb_e=\mathbf{0}$, $\bb_u=\mathbf{0}$} is covered by~\eqref{eq:compatibility_1} and~\eqref{eq:compatibility_2}.

\emph{Anti-parallel case (${\norm{\bb_e}\norm{\bb_u}\!+\!\bb_e^\top\bb_u=0}$, while ${\bb_e\ne\mathbf{0}}$, ${\bb_u\ne\mathbf{0}}$):}
We have ${\bb_e \!=\! -c\,\bb_u}$ with
${c\!=\!\norm{\bb_e}/\norm{\bb_u}>0}$.
So ${\lambda_e\bb_e\!+\!\lambda_u\bb_u=\mathbf{0}}$ gives ${\lambda_u = c \lambda_e}$, and
\eqref{eq:compat_farkas} yields~\eqref{eq:compatibility_3}.
%
\end{proof}

Conditions~\eqref{eq:compatibility_1}--\eqref{eq:compatibility_2} coincide with the individual robust ICBF conditions~\eqref{eq:tissf_rd} for $h_e$ and $h_u$. 
Hence, they are satisfied when $h_e$ and $h_u$ are ICBFs.
However, two ICBFs need not form a compatible ICBF pair: condition~\eqref{eq:compatibility_3}
may fail when $\bb_e$ and $\bb_u$ are anti-parallel.
Furthermore, even if compatibility holds without robustness terms, it may not hold with ${r_e,r_u>0}$: unlike the single-constraint case, a compatible ICBF pair is not automatically a compatible robust ICBF pair.
Condition~\eqref{eq:compatibility_3} must therefore be explicitly verified or enforced through the design of $\phib$.
Tuning $\phib$ directly shapes $a_e$ and $a_u$ and can be used to ensure compatibility.


\subsection{Robust Safety}

Having established feasibility, we now state the formal safety guarantees provided by the proposed controller.
Since the input history over $[0,\tau]$ is already committed and cannot be altered by the safety filter, we require an assumption.
\begin{assumption}
\label{assump:history}
The initial input history satisfies $(\x(\vartheta),\, \uu(\vartheta{-}\tau)) \in \Scal_i$ for all $\vartheta \in [0,\tau]$.
\end{assumption}
This requirement is stronger than $(\x(0),\uu(0))\in\Scal$ for the delay-free case: with delay, the safety filter has no authority over $[0,\tau]$ since the input history is already committed.

Under Assumptions~\ref{assump:reg_value}--\ref{assump:history}, we establish the safety of the time delay system \eqref{eq:plant_delay} with the integral controller~\eqref{eq:ctrl_disturbed} in the presence of prediction errors.
Specifically, we show that the set $\Scal_i$ is forward invariant if the disturbance caused by prediction errors is small (i.e., ${\delta \leq \mu_i(0)}$).
Moreover, even if the disturbance is large (i.e., ${\delta > \mu_i(0)}$), we can show that trajectories remain within a neighborhood of the safe set by using the concept of input-to-state safety~\cite{ames2019issf}.
For each barrier $h_i$, we define the inflated set
\begin{equation}
  \Scal_{i,\delta} := \big\{(\x,\uu) \in \mathbb{R}^n \!\times \mathbb{R}^m :
    h_{i,\delta}(\x,\uu) \ge 0\big\},
  \label{eq:safeset_delta_pair}
\end{equation}
where
\begin{equation}
  h_{i,\delta}(\x,\uu):= h_i(\x,\uu) - \alpha_i^{-1}\!\left(-\frac{(\mu_i(0) - \delta)^2}{4\,\sigma_i(h_i(\x,\uu))}\right).
  \label{eq:safeset_delta_pair_h}
\end{equation}
Note that the set $\Scal_{i,\delta}$ grows with $\delta$ (if ${\delta > \mu(0)}$), while $\delta$ itself grows with the prediction error ${\xph - \xp}$.
A more accurate prediction, therefore, directly reduces conservatism~\cite{molnar2022input}.

The following theorem establishes the forward invariance of $\Scal_i$ and $\Scal_{i,\delta}$ for both the single- and two-constraint cases.

\begin{theorem}
\label{thm:tissf}
Consider the system~\eqref{eq:plant_delay},
the predicted state $\xp$ in \eqref{eq:semiflow},
its estimation $\xph$,
the safe set $\Scal_i$ in \eqref{eq:safeset_delay},
and the inflated set $\Scal_{i,\delta}$ in \eqref{eq:safeset_delta_pair}.
Let Assumptions~\ref{assump:reg_value},~\ref{assump:delta}, and~\ref{assump:history} hold.

If the condition in \eqref {eq:tissf_rd} holds,
then any locally Lipschitz controller ${\vv=\qq(\xp,\uu)}$, with ${\qq : \mathbb{R}^n \!\times\! \mathbb{R}^m \to \mathbb{R}^m}$, satisfying
\begin{equation}
  \bb^\top(\xp,\uu)\vv
  \ge a(\xp,\uu) + r(\xp,\uu)
  \label{eq:tissf_cond}
\end{equation}
for all ${(\xp,\uu)\in\Scal_\delta}$,
renders $\Scal_\delta$ forward invariant for \eqref{eq:plant_delay}--\eqref{eq:ctrl_disturbed}.
Further, if ${\delta \leq \mu(0)}$, then $\Scal$ is also forward invariant,
i.e.,
${(\x(0),\uu(-\tau))\in\Scal \implies (\x(t),\uu(t{-}\tau))\in\Scal}$, ${\forall\,t\ge 0}$.

If the conditions in \eqref{eq:robust_compat} hold,
then any locally Lipschitz controller ${\vv=\qq(\xp,\uu)}$, with ${\qq : \mathbb{R}^n \!\times\! \mathbb{R}^m \to \mathbb{R}^m}$, satisfying
\begin{equation}
  \bb_i^\top(\xp,\uu)\vv
  \ge a_i(\xp,\uu) + r_i(\xp,\uu),
  \quad \forall i\in\{e,u\},
  \label{eq:tissf_cond2}
\end{equation}
for all ${(\xp,\uu)\in \Scal_{e,\delta}\cap\Scal_{u,\delta}}$,
renders $\Scal_{e,\delta}\cap\Scal_{u,\delta}$ forward invariant for \eqref{eq:plant_delay}--\eqref{eq:ctrl_disturbed}.
Further, if ${\delta \leq \min\{\mu_e(0),\mu_u(0)\}}$, then ${\Scal_e\cap\Scal_u}$ is also forward invariant,
i.e.,
$(\x(0),\uu(-\tau))\in\Scal_{e}\cap\Scal_{u} \implies (\x(t),\uu(t{-}\tau)) \in \Scal_{e}\cap\Scal_{u},\ \forall\,t\ge 0$.
\end{theorem}


\begin{proof}
We follow~\cite[Appendix]{molnar2022safety},\cite{Alan2022} and give the argument for a single barrier $h$; the two-barrier case follows by applying the same reasoning to $h_e$ and $h_u$.
Shifting time by $\tau$, the predicted pair $(\xp,\uu)$ obeys the delay-free dynamics
${\dot{\x}_{\rm p}(t) = \f(\xp(t),\uu(t))}$,
with ${\dot{\uu} = \phi(\xp,\uu) + \qq(\xp,\uu) + \mathbf{d}}$ based on \eqref{eq:ctrl_disturbed}.
We prove forward invariance using Nagumo's theorem.
First, we compute $\dot{h}$ along the disturbed dynamics
\begin{align}
    & \dot{h} \big( \xp,\uu,\qq(\xp,\uu)+\mathbf{d} \big) \nonumber \\
    & = - \alpha \big( h(\xp,\uu) \big) - a(\xp,\uu) + \bb^\top(\xp,\uu) \big( \qq(\xp,\uu) + \mathbf{d} \big). \nonumber
\end{align}
For ease of notation, we omit the argument of $\dot{h}$ and the arguments $(\xp,\uu)$ from now on.
Substituting~\eqref{eq:tissf_cond}, we get
\begin{align}
    & \dot{h} \geq - \alpha \big( h \big) + r + \bb^\top \mathbf{d}. \nonumber
\end{align}
Then, we substitute $r$ from~\eqref{eq:r}, and we write ${\bb^\top\mathbf{d} \!\ge\! -\norm{\bb}\delta}$ based on the Cauchy--Schwarz inequality and Assumption~\ref{assump:delta}:
\begin{align}
    \dot{h} \geq - \alpha(h) + \big( \mu (h) - \delta \big) \|\bb\| + \sigma(h) \|\bb\|^2. \label{eq:proof_1}
\end{align}
Further, by completing the square, it can be shown that:
\begin{align}
    \dot{h} \geq - \alpha(h) + \big( \mu(h) - \mu(0) \big) \|\bb\| - \frac{(\mu(0)-\delta)^2}{4 \sigma(h)}. \label{eq:proof_2}
\end{align}

We prove the forward invariance of $\Scal$ by considering ${h = 0}$.
Substituting into \eqref{eq:proof_1} and using  ${\delta \leq \mu(0)}$:
\begin{align}
    & \dot{h} \geq \big( \mu(0) - \delta \big) \|\bb\| + \sigma(0) \|\bb\|^2 \geq 0. \nonumber
\end{align}
Since zero is a regular value of $h$ based on Assumption~\ref{assump:reg_value}, Nagumo's theorem yields forward invariance for $\Scal$. 
Similarly, to show the invariance of $\Scal_\delta$, we take ${h_\delta = 0}$.
This implies
${h \leq 0}$ and
${\alpha(h) \!+\! \frac{(\mu(0)-\delta)^2}{4 \sigma(h)} \!=\! 0}$,
thus \eqref{eq:proof_2} yields
\begin{align}
    & \dot{h} \geq \big( \mu(h) - \mu(0) \big) \|\bb\| \geq 0, \nonumber
\end{align}
where ${\mu(h) \geq \mu(0)}$ as $\mu$ is monotonously decreasing.
Based on~\cite[proof of Thm.~3]{Alan2022}, ${\dot{h}_\delta \geq 0 \iff \dot{h} \geq 0}$, and zero is a regular value of $h_\delta$ via Assumption~\ref{assump:reg_value}.
Hence, Nagumo's theorem yields forward invariance for $\Scal_\delta$.
These give safety for ${t\ge\tau}$, while Assumption~\ref{assump:history} covers ${t\in[0,\tau]}$.
\end{proof}




\begin{remark}[Special Cases]
In case of accurate prediction (${\xph = \xp}$),
 one could set ${r_i(\xp,\uu)=0}$,
reducing~\eqref{eq:robust_cond} to the standard ICBF conditions evaluated at the predicted state $\xp$:
\begin{equation}
  \bb_i^\top(\xp,\uu)\vv \ge a_i(\xp,\uu).
  \label{eq:icbfd_safety_cond}
\end{equation}
The corresponding minimum-norm safety filter is:
\begin{align}
  \qq(\xp,\uu) = \argmin_{\vv\in\mathbb{R}^m} &\quad \norm{\vv}^2 \\ 
  \mathrm{s.t.} & \quad 
  \bb_i^\top(\xp,\uu)\vv\ge a_i(\xp,\uu), \nonumber
  \label{eq:icbfd_qp}
\end{align}
cf.~\eqref{eq:combined_qp_delay}.
%
Additionally, in case of zero delay (${\tau=0}$), we have ${\xp=\x}$, 
thus Assumption~\ref{assump:history} holds automatically, and 
\eqref{eq:icbfd_safety_cond}
reduces to the delay-free ICBF condition in 
Section~\ref{sec:background}.
\end{remark}

\section{Adaptive Cruise Control with Delay}
\label{sec:example}

We illustrate the proposed framework on an adaptive cruise control (ACC) system \cite{ames2017cbf}.
Consider an automated vehicle (AV) following a lead vehicle.
The AV's state, $\x=[D~v]^\top$, includes the distance to the lead vehicle, $D$, and the AV's speed, $v$.
With actuation delay $\tau>0$, the dynamics are
\begin{equation}
  \dot{D} = v_{\rm L} - v, \qquad
  \dot{v} = u(t{-}\tau) - p(v),
  \label{eq:acc}
\end{equation}
where $v_{\rm L}$ is the lead vehicle's speed,
$u$ is the commanded acceleration,
and ${p(v)=c_0+c_1 v+c_2 v^2}$ is a resistance term.

\begin{figure}[t]
  \centering
  \includegraphics[width=1\columnwidth]{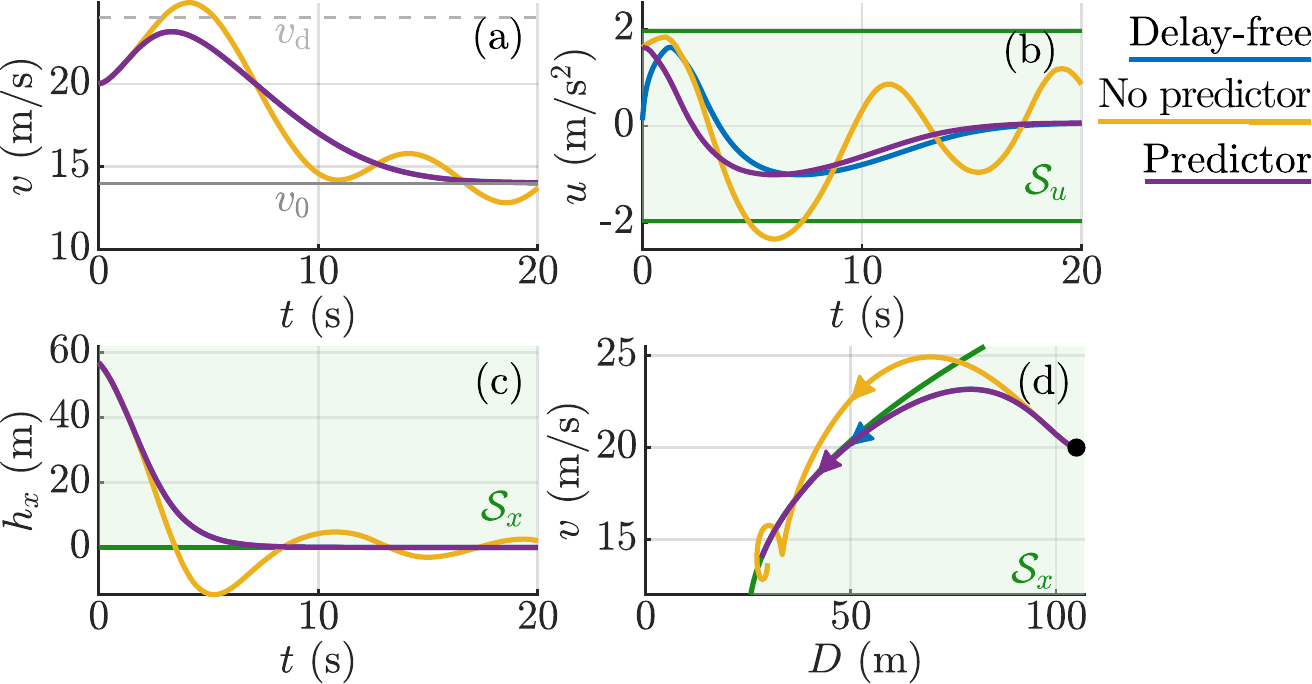}
  \vspace{-5mm}
  \caption{%
    Simulation of ACC, showing the effect of input delay and predictor compensation: (a)~velocity; (b)~acceleration; (c)~CBF $h_x$; (d)~phase portrait. Blue: delay-free reference; yellow: naive controller unaware of input delay; purple: the proposed predictor feedback.
    }
  \label{fig:predictor}
  \vspace{-2mm}
\end{figure}

Our goal is to control the AV such that it maintains a safe distance $D$ while keeping the acceleration command $u$ within safe bounds. 
We use the proposed controller in \eqref{eq:ctrl_disturbed} and \eqref{eq:combined_qp_delay} with
the dynamically defined controller $\phib(\xp,u)$ in \eqref{eq:phi_tracking_delay} and the nominal controller
${k_{\rm d}(\x) = K_v\,(v_{\rm d} - v)}$, which drives $v$ towards a desired speed $v_{\rm d}$ with gain $K_v$.
The safe following distance is encoded by the barrier
\begin{equation}
  h_x(\x) = D - T_{\rm h}\,v - \tfrac{(v_{\rm L}-v)^2}{2\,u_{\max}} - D_{\rm sf},
  \label{eq:acc_hx}
\end{equation}
where $T_{\rm h}$ is the time headway, $u_{\max}$ is the maximum deceleration, and $D_{\rm sf}$ is a safe standstill distance.
This is used to construct the extended barrier $h_e$ in \eqref{eq:he}.
The acceleration constraint ${|u|\le u_{\max}}$ is captured by
\begin{equation}
  h_u(u) = u_{\max}^2 - u^2.
  \label{eq:acc_hu}
\end{equation}

Feasibility of the QP \eqref{eq:combined_qp_delay} is analyzed via Proposition~\ref{prop:compat}. 
Conditions~\eqref{eq:compatibility_1} and~\eqref{eq:compatibility_2} are the individual ICBF conditions for $h_e$ and $h_u$ separately. 
Condition~\eqref{eq:compatibility_3} is nontrivial and requires attention when $b_e$ and $b_u$ have opposite signs, which occurs during braking (${u < 0}$, so ${b_u = -2u > 0}$) when ${b_e = -(T_{\rm h} + (v_{\rm L}-v)/u_{\max}) < 0}$. 
In this regime, the condition~\eqref{eq:compatibility_3} must hold. Without the braking-distance term $(v_{\rm L}-v)^2/(2u_{\max})$ in \eqref{eq:acc_hx}, $b_e = -T_{\rm h}$ would be constant and negative for all $v$, so opposite signs -- and thus condition~\eqref{eq:compatibility_3} -- would need to be checked throughout the entire braking phase with no guarantee of satisfaction. 
The quadratic term introduces a state-dependent contribution to $b_e$ \eqref{eq:acc_hx} that reduces the region where condition~\eqref{eq:compatibility_3} is active, and encodes the physically correct stopping distance.

\begin{figure}[t]
  \centering
  \includegraphics[width=1\columnwidth]{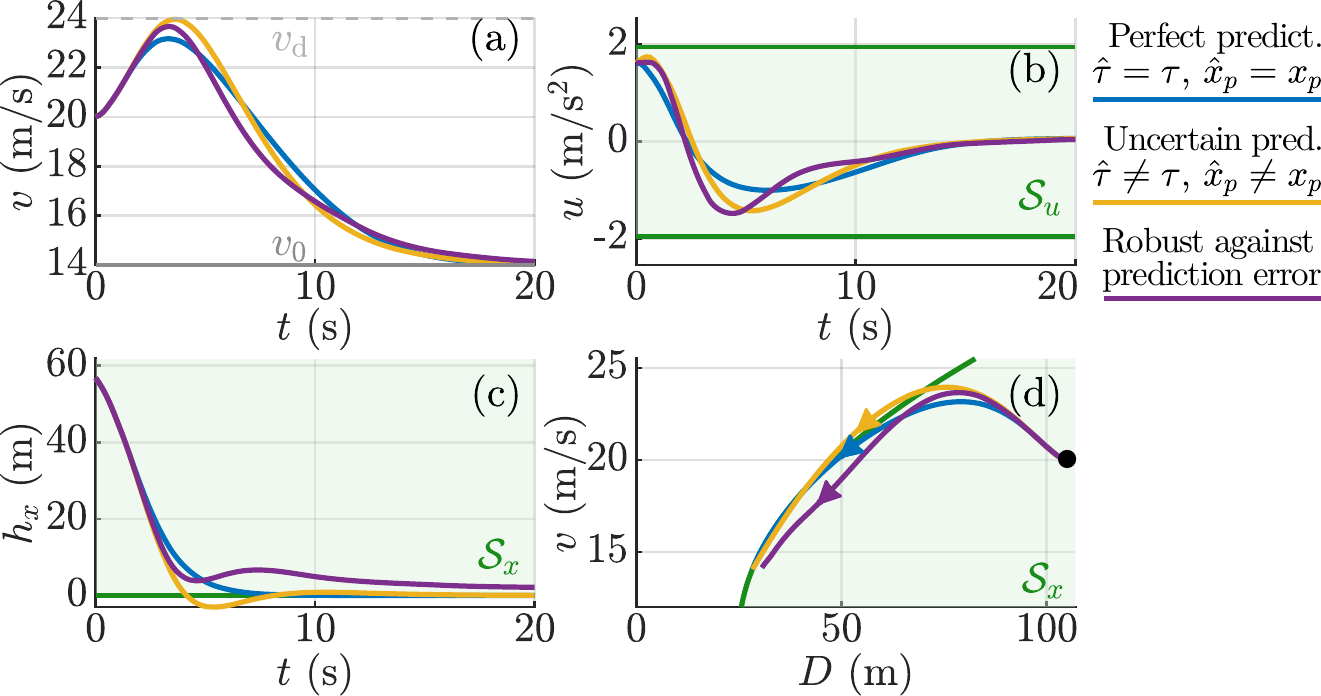}
  \vspace{-5mm}
  \caption{%
    Simulation of ACC, demonstrating robustness to prediction error (with $50\%$ delay underestimation): (a)~velocity; (b)~acceleration; (c)~CBF $h_x$; (d)~phase portrait. Blue: perfect prediction; yellow: uncertain prediction and no robustification; purple: uncertain prediction with robustification. 
  }
  \label{fig:tissf}
  \vspace{-7mm}
\end{figure}

\begin{table}
\begin{center}
\caption{ACC simulation parameters.}
\label{tab:acc_params}
\footnotesize
\begin{tabular}{cccc}
Parameter & Symbol & Value & Unit \\
\hline
\multirow{3}{*}{resistance} 
  & $c_0$ & $6.06\times10^{-5}$ & m/s$^2$ \\
  & $c_1$ & $3.03\times10^{-3}$ & s$^{-1}$ \\
  & $c_2$ & $1.52\times10^{-4}$ & m$^{-1}$ \\
  vel.\ tracking gain      & $K_v$                        & $1$                & s$^{-1}$ \\
integral control gain       & $\alpha_\phi$                & $3$                & s$^{-1}$ \\
\hline
R-CBF linear             & $\mu_{0,e},\,\mu_{0,u}$      & $1.0,\;0.2$        & m/s$^3$ \\
R-CBF quadratic          & $\sigma_{0,e},\,\sigma_{0,u}$& $0.1,\;0.05$       & m/s$^4$;\;s$^{-1}$ \\
R-CBF decay              & $\lambda$                    & $0.05$             & s/m \\
\end{tabular}
\vspace{-8mm}
\end{center}
\end{table}

The two simulation studies below use the ICBFs \eqref{eq:acc_hx}--\eqref{eq:acc_hu} and the QP \eqref{eq:combined_qp_delay}.
In the first case, we show the effect of the delay and the predictor ${\xp=\Psi(\tau,\x,\uh)}$.
In the second case, we highlight what happens when the delay is uncertain, ${\hat\tau\ne\tau}$,  and the predictor ${\xph=\Psi(\hat\tau,\x,\uh)}$ introduces a disturbance.
The second case uses the robust margins~\eqref{eq:r} with ${\mu_i(h_i)=\mu_{0,i} e^{-\lambda h_i}}$, ${\sigma_i(h_i)=\sigma_{0,i} e^{-\lambda h_i}}$, while ${\alpha_i(h)=\gamma_i h_i}$ is applied in the QP \eqref{eq:combined_qp_delay} in each case.
The simulation parameters are ${u_{\max}=1.96\,\rm m/s^2 }$, ${\tau=1.2\,\rm s}$, ${\hat\tau=0.6\,\rm s}$, ${D(0)=105\,\rm m}$, , ${v(0)=20\,\rm m/s}$, ${v_{\rm L}(t) \equiv14\,\rm m/s}$, ${v_{\rm d}=24\,\rm m/s}$, ${\gamma_x=\gamma_e=\gamma_u=1\,\rm s^{-1}}$, ${T_{\rm h}=1.8\,\rm s}$, and ${D_{\rm sf}=3\,\rm m}$; while the remaining parameters are in Table~\ref{tab:acc_params}.

Figure~\ref{fig:predictor} shows the effect of input delay and predictor feedback.
The delay-free case (blue) serves as reference with state and input constraints maintained throughout.
With delay and without predictor (yellow), the ICBF filter is applied naively at the current state $\x$, unaware of the delay.
This controller fails to anticipate the required braking and both constraints are violated.
With predictor feedback (purple), the controller fully compensates the delay: the acceleration in panel~(b) is a shifted copy of the delay-free case, while the velocity and barrier in panels~(a) and (c) recover the delay-free response.
Thus safety is maintained per Theorem~\ref{thm:tissf}.

Figure~\ref{fig:tissf} shows the effect of delay uncertainty, when the delay is underestimated by  $50\%$ (${\hat\tau=\tau/2}$).
With perfect prediction (blue, repeated from Fig.~\ref{fig:predictor}), safety is maintained exactly.
With uncertain prediction and no robustification (yellow), the disturbance $\mathbf{d}$ causes the violation of the state constraint.
As opposed, the proposed robust ICBF filter (purple) intervenes proactively and maintains both constraints during the simulation.

\section{Conclusion}
\label{sec:conclusion}

This paper developed an ICBF framework for the safety-critical control
of nonlinear systems with input delay and joint state and input constraints. Three contributions were presented. 
First, predictor feedback was incorporated so that the safety filter compensates for the delay.
Second, a closed-form feasibility condition was derived for the two-constraint QP arising from simultaneous state and input constraints.
Third, a robust ICBF formulation was proposed to address prediction errors caused by, for example, delay uncertainty.
The framework was validated on an adaptive cruise control example, where delay compensation and robustness to a $50\%$ delay estimation error were demonstrated in simulation.
Future research directions include experimental validation on embedded platforms with online delay estimation.







\bibliographystyle{IEEEtran}
\bibliography{references}	

\end{document}